%% file: orb.tex
\documentclass[12pt]{elsart}

\usepackage{amssymb,amsmath} 
\usepackage{epsfig,float}

\def\jmod#1 #2{#1\ ({\rm mod}\ #2)}
\def\prend{\vrule depth-1pt height7pt width6pt}
\def \endpf{{\hfill\prend\bigbreak}}

\newtheorem{theorem}{Theorem}
\newtheorem{lemma}[theorem]{Lemma}
\newtheorem{corollary}[theorem]{Corollary}

\newtheorem{examp}[theorem]{Example}
\newenvironment{example}{\begin{examp}\normalfont\quad}{\normalsize\end{examp}}

\begin{document}

\begin{frontmatter}

\title{Periodicity, Repetitions, and Orbits of an Automatic Sequence}

\author[France]{Jean-Paul Allouche},
\ead{allouche@lri.fr}
\author[Winnipeg]{Narad Rampersad},
\ead{n.rampersad@uwinnipeg.ca}
\author[Waterloo]{Jeffrey Shallit\corauthref{cor}}
\corauth[cor]{Corresponding author.}
\ead{shallit@cs.uwaterloo.ca}
\address[France]{CNRS, LRI, UMR 8623, Universit\'e Paris-Sud, B\^atiment 490, F-91405 Orsay Cedex, France}
\address[Winnipeg]{Department of Mathematics and Statistics, University of Winnipeg, 515 Portage Avenue, Winnipeg, MB R3B 2E9, Canada}
\address[Waterloo]{School of Computer Science, University of Waterloo, Waterloo, Ontario  N2L 3G1, Canada}

\begin{abstract}
We revisit a technique of S. Lehr on automata and use it to prove old
and new results in a simple way.  We give a very simple proof of the
1986 theorem of Honkala that it is decidable whether a given
$k$-automatic sequence is ultimately periodic.  We prove that it is
decidable whether a given $k$-automatic sequence is overlap-free (or
squarefree, or cubefree, etc.) We prove that the lexicographically
least sequence in the orbit closure of a $k$-automatic sequence is
$k$-automatic, and use this last result to show that several related
quantities, such as the critical exponent, irrationality measure, and
recurrence quotient for Sturmian words with slope
$\alpha$, have automatic
continued fraction expansions if $\alpha$ does.
\end{abstract}


\end{frontmatter}

\centerline{\it In Honor of Juhani Karhum\"aki's 60th Birthday}

\def\orb{{\rm Orb}}
\def\cl{{\rm Cl}}

\section{Introduction}
\label{intro-sec}

A sequence $(a_n)_{n\geq 0}$ over a finite alphabet $\Delta$ is said to
be {\it $k$-automatic} for some integer $k \geq 2$ if, roughly
speaking, there exists an automaton that, on input $n$ in base $k$,
reaches a state with the output $a_n$.
More formally, a sequence $(a_n)_{n\geq 0}$ over $\Delta$
is $k$-automatic if there exists a 
deterministic finite automaton with output (DFAO) $M = (Q, \Sigma_k,
\Delta, \delta, q_0, \tau)$ where $Q$ is a finite set of states,
$\Sigma_k = \lbrace {\tt 0,1,2,}\ldots, k-1 \rbrace$,
$\delta : Q \times \Sigma_k \rightarrow Q$ is the transition function,
and $\tau:Q \rightarrow \Delta$ is the output function, such that if
$w$ is any base-$k$ representation of $n$, possibly with leading zeroes,
then $a_n = \tau(\delta(q_0, w^R))$.  (Note that $a_0 = \tau(q_0)$.)
Here $w^R$ is the reverse of the word $w$.

This class of sequences, also
called $k$-recognizable in the literature, has been studied extensively
(e.g., \cite{Allouche&Shallit:2003}) and
has several different characterizations, the most
famous being images (under a coding)
of fixed points of $k$-uniform morphisms.

The archetypal example of a $k$-automatic sequence is the {\it Thue-Morse
sequence} 
$${\bf t} = (t_n)_{n \geq 0} = {\tt 0110100110010110} \cdots,$$
where $t_n$ is the sum (modulo $2$) of the bits in the base-$2$ expansion
of $n$ \cite{Allouche&Shallit:1999}.  See Figure~\ref{tm}.
It can also be viewed as the
fixed point of the morphism $\mu$ where ${\tt 0} \rightarrow {\tt 01}$
and ${\tt 1} \rightarrow {\tt 10}$.

\begin{figure}[H]
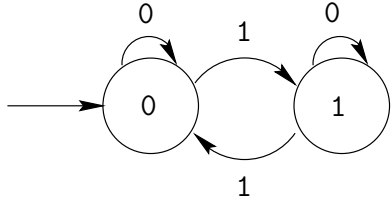

\begin{center}
\input orb2.tex
\end{center}
\caption{Automaton generating the Thue-Morse sequence}
\label{tm}
\end{figure}

Given a $k$-automatic sequence, one might reasonably inquire as to whether the
sequence is ultimately periodic.  More precisely, we would like to know if
the problem

\centerline{Given a $k$-automatic sequence, is it ultimately periodic?}

\noindent is decidable (i.e., recursively solvable).
This problem was solved by Honkala \cite{Honkala:1986},
who gave a rather complicated decision procedure.

In this paper, we begin by recalling a technique of Lehr
\cite{Lehr:1993}
as simplified by Allouche and Shallit
\cite[pp.\ 380--382]{Allouche&Shallit:2003}.  
In Section~\ref{ultimate} we introduce it and use it
to reprove the result of Honkala mentioned above.

Another topic of great interest is the pattern-avoiding properties of 
certain automatic sequences.  For example, more than a hundred years ago
Thue proved \cite{Thue:1906,Thue:1912} that $\bf t$ contains no overlaps,
where an {\it overlap} is a word of the form $axaxa$, where $a$ is a single
letter and $x$ is a word, possibly empty.    Examples of overlaps include
{\tt alfalfa} in English, {\tt entente} in French, and 
{\tt ajaja} and {\tt tutut} in Finnish.

Similarly, much attention has been given to avoiding squares.  A {\it square} is
a word of the form $xx$ where $x$ is nonempty.  
Examples of squares include
{\tt murmur} in English, {\tt chercher} in French, and {\tt valtavalta} in Finnish.
A (finite or infinite) word is {\it squarefree} if it contains no square
factor.  As is well-known, if one counts the lengths
of the blocks of $1$'s between consecutive $0$'s in $\bf t$,
one obtains the squarefree sequence
$${\bf v} = (v_n)_{n \geq 0} = {\tt 210201210120} \cdots  .$$
The word {\bf v} is generated as the fixed point of the morphism $g$ 
defined by
${\tt 2} \rightarrow {\tt 210}$, ${\tt 1} \rightarrow {\tt 20}$,
and ${\tt 0} \rightarrow {\tt 1}$.  
Furthermore, $\bf v$ is generated by the automaton depicted
in Figure~\ref{sqf}.  Here the input is $n$ expressed in base $2$,
starting with the least significant digit, and the output, given by the symbol
labeling the state, is $v_n$.  (Contrast this with the representation
given by Berstel \cite{Berstel:1978}.)

\begin{figure}[H]
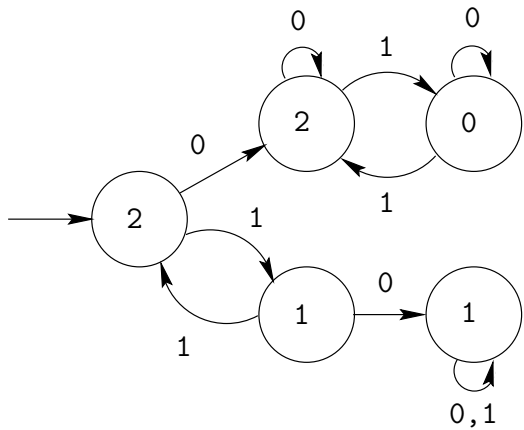

\begin{center}
\input orb1.tex
\end{center}
\caption{Automaton generating a squarefree sequence}
\label{sqf}
\end{figure}

We can generalize the concept of power to non-integer powers.  Let $\alpha$
be a real number $> 1$.  We say that
a word $z$ is an {\it $\alpha$-power} if it is the shortest
prefix of length $\geq \alpha |x|$
of some infinite word $x^\omega = xxx \cdots$, and we say it is an 
{\it $\alpha^+$-power}
if it is the shortest prefix of length $> \alpha |x|$ of $x^\omega$.  
For example,  the English word
$z = {\tt abracadabra}$ is both a $3/2$ and a $(3/2)^+$ power, as 
$z$ is a prefix of length $11$ of $({\tt abracad})^\omega$, and
$10/7 < 3/2 < 11/7$.  Using this notation, an overlap is a $2^+$ power.
We say a (finite or infinite) word $z$ contains
an $\alpha$-power if we can write $z = uvw$ where $v$ is an $\alpha$-power.
We say that a (finite or infinite) word $z$ {\it avoids $\alpha$-powers}
or is {\it $\alpha$-power-free}
if it has no factor that is an $\alpha$-power, and similarly for
$\alpha^+$-powers.

In Section~\ref{squarefree} we use Lehr's technique to prove a new result:
that it is decidable whether a given $k$-automatic sequence is squarefree,
overlap-free, contains an $r$-power for $r$ rational,
contains an $r^+$-power, etc.

Let ${\bf a} = (a_n)_{n \geq 0}$ be a sequence over a finite alphabet
$\Delta$.  The {\it orbit} of $\bf a$, written $\orb ({\bf a})$, is the
set of all its shifts, that is, the set of sequences
$\lbrace (a_{n+i})_{n \geq 0} \ : \ i \geq 0 \rbrace$.
The {\it orbit closure} of $\bf a$, written $\cl( \orb ({\bf a}))$ is the
closure of $\orb ({\bf a})$ under the usual topology where two sequences
are close if they agree on a long prefix.  More transparently, a sequence
${\bf b} = (b_n)_{n \geq 0}$ is in the orbit closure of $\bf a$ if and only
if every finite prefix of $\bf b$ is a factor of $\bf a$ \cite[Prop.\ 10.8.9,
p.\ 327]{Allouche&Shallit:2003}.

An infinite word $\bf a$ is said to be {\it recurrent} if every finite factor
that occurs in $\bf a$ occurs infinitely often.  It is not hard to see that
if $\bf a$ is recurrent and not periodic, then $\cl( \orb ({\bf a}))$
is uncountable \cite[Thm.\ 10.8.12, p.\ 328]{Allouche&Shallit:2003}.  If
$\bf a$ is not recurrent this may not be true; for example, consider
the infinite word ${\bf c} = {\tt abaabaaabaaaab} \cdots$.  Then
$\cl( \orb ({\bf c}))$ is countable because once a finite factor contains
two or more $\tt b$'s, its position in $\bf c$ is fixed and hence can be
extended in at most one way.  Thus $\cl( \orb ({\bf c}))$ equals
${\tt a}^\omega \ \cup \ {\tt a}^* {\tt b} {\tt a}^\omega \ \cup \ 
\orb ({\bf c})$, and hence is countable.

In Section~\ref{orbit} we are interested in elements in the orbit closure of
automatic sequences.   From the result mentioned above, if $\bf a$ is
recurrent, then ``most'' of the sequences in $\cl( \orb ({\bf a}))$ cannot
be $k$-automatic for any $k$, since the orbit closure is uncountable while
the set of $k$-automatic sequences over $\Delta$ is countable.  Evidently,
this is true even if $\bf a$ itself is not automatic.  

Now suppose that $\bf a$ is $k$-automatic, and consider the lexicographically
least sequence $\bf b$ in $\cl( \orb ({\bf a}))$.  We show in Section~\ref{orbit} that
$\bf b$ is also $k$-automatic, and more generally, any sequence chosen
in a periodic way from the factor tree of $\bf a$ is also $k$-automatic.

\section{Periodicity}
\label{ultimate}

Let ${\bf a} = (a_n)_{n \geq 0}$ be an infinite sequence.  Then
$\bf a$ is {\it ultimately periodic} if there exist integers
$P \geq 1, N \geq 0$ such that $a_i = a_{i+P}$ for all $i \geq N$.

\bigskip

\begin{theorem}
Given a DFAO
$M= (Q, \Sigma_k, \Delta, \delta, q_0, \tau)$
it is decidable if the $k$-automatic sequence it generates is
ultimately periodic.
\label{periodic}
\end{theorem}

    As mentioned before, this result is due to Honkala \cite{Honkala:1986}.
We give a new proof.

\begin{proof}
We start with a sketch of the proof.  First, we construct an NFA $M_1$ that
on input $(P,N)$ ``guesses'' $I$ and accepts if $I \geq N$ and
$a_I \not= a_{I+P}$.  We now convert $M_1$ to a DFA $M_2$ using the
usual subset construction, and then interchange accepting and non-accepting
states, obtaining a DFA $M_3$ with the property that $M_3$ accepts
$(P,N)$ if and only if $a_I = a_{I+P}$ for all $I \geq N$.  Now $\bf a$
is ultimately periodic if and only if $M_3$ accepts some input, which can
be checked using the usual depth-first search technique to determine if there
is a path from $M_3$'s initial state to a final state.

We now give the proof in detail, addressing concerns such as exactly how
$P$ and $N$ are represented, what it means to guess $I$,
how we verify that $I \geq N$, how we compute
$I+P$, and what if $I$ is significantly larger than $P$ or $N$.

When we say that $M_1$ takes $(P,N)$ as input, what we really mean is that
the input alphabet of $M_1$ is $\Sigma_k \times \Sigma_k$, so that 
$M_1$ takes as input the base-$k$ digits of $P$ and $N$ {\it in parallel}.
More precisely, the input is $(p_0, n_0)(p_1, n_1) \cdots (p_j, n_j)$ where
$n_j n_{j-1} \cdots n_0$ is a base-$k$ representation of $N$ and
$p_j p_{j-1} \cdots p_0$ is a base-$k$ representation of $P$, either or
both padded with leading zeros to ensure that their lengths are the same.
This means that $(P,N)$ can be input in infinitely many ways, depending on the
number of leading zeros (which are actually trailing zeros since we read the input starting with the least significant digit),
and we must ensure that the correct result is returned in each case.  

When we say we guess $I$, what we really mean is that we 
successively guess the base-$k$ digits of $I$, starting with the least
significant digit.

In order to verify that our guessed $I$ is $\geq N$, we maintain a flag
that records how the number represented by the digits of $I$ seen so far
stands in relation to the digits of $N$ seen so far:  whether it is
$<$, $=$, or $>$.  The flag is updated as follows, if the next digit of
$I$ guessed is $i'$ and the next digit of $N$ is $n'$:

\begin{eqnarray}
u(<,i', n') &=&  
	\begin{cases}
	<,  & \text{if } i' \leq n' ; \\
	>,  & \text{if } i' > n';
	\end{cases} \nonumber \\
u(=,i',n') &=& 
	\begin{cases}
	<, & \text{if } i' < n' ; \\
	=, & \text{if } i' = n' ; \\
	>, & \text{if } i' > n' ;
	\end{cases} \label{update} \\
u(>,i',n') &=& 
	\begin{cases}
	<, & \text{if } i' < n' ; \\
	>, & \text{if } i' \geq n'.
	\end{cases} \nonumber
\end{eqnarray}

To compute $I+P$, we maintain a ``carry'' bit, and compute $I+P$ digit-by-digit
as we see the digits of $P$ input using the usual pencil-and-paper method.

Finally, since we guess the digits of $I$ in parallel with the digits
of the inputs $P$ and $N$, we have to address the situation where the
base-$k$ representation of the appropriate $I$ to guess is longer than
the representation of the inputs $P$ and $N$.  If we do not pad $P$ and
$N$ with enough $0$'s, we might 
return the wrong result.  To handle this, we modify the acceptance criterion
of the NFA $M_1$,  making a state accepting 
if an accepting state could  be reached by any input of the
form $(0,0)^j$, $j \geq 0$.    

We now give the construction in more detail.  Suppose
$M = (Q, \Sigma_k, \Delta, \delta, q_0, \tau)$ is a $k$-DFAO.  
We make an NFA $M_1 = (Q', \Sigma_k \times \Sigma_k, \delta', q'_0, F')$
as follows.   

\begin{eqnarray*}
Q' &= & \lbrace <, =, > \rbrace \times \lbrace 0,1 \rbrace \times 
Q \times Q ; \\
q'_0 & = & [=, 0, q_0, q_0]; \\ 
F' & = & \lbrace [b,0,q,r] \ : \ b \in \lbrace >, = \rbrace \text{ and }
	\tau(q) \not= \tau(r) \rbrace \\
\end{eqnarray*}

The meaning of a state $[b,c,q,r]$ of $Q'$ is that $b$ is the flag
maintaining the relationship between $I$ and $N$; $c$ is the carry bit
in the computation of $I+P$; $q$ is the state in $M$ reached by the bits
of $I$ seen so far; and $r$ is the state in $M$ reached by the bits of
$I+P$ calculated so far.

We define $\delta'$ by
$\delta'([b,c,q,r], (n', p')) :=$
$$
\lbrace [u(b,i',n'), \lfloor {{i'+p'+c} \over k} \rfloor,
\delta(q,i), \delta(r, (i'+p'+c) \bmod k) ] \ : \ 0 \leq i' < k \rbrace.$$
Here $u$ is the update map defined in Eq.~(\ref{update}).  

This finishes the construction of the NFA $M_1$.  We now create a new
NFA $M'_1$ that is exactly the same as $M_1$, except that it has a new
set of final states $\hat{F'}$ defined by
$$\hat{F'} := \lbrace [b,c,q,r] \ : \ \text{ there exists } j \geq 0
\text{ such that } \delta'([b,c,q,r], (0,0)^j ) \in F' \rbrace.$$

We now convert $M'_1$ to a DFA $M_2 = (Q'', \Sigma_k\times \Sigma_k,
\delta'', q''_0, F'')$
using the usual subset construction.  We define
$M_3 = (Q'', \Sigma_k\times \Sigma_k,
\delta'', q''_0, Q''-F'').$  It is not hard to see that 
$M_3$ accepts some input $(P,N)$ with $P \geq 1$
if and only if $\bf a$ is ultimately periodic.
This can be checked by creating a DFA $M_4$ that
accepts $(\Sigma_k^* (\Sigma_k - \lbrace {\tt 0} \rbrace)\Sigma_k^*) \times
\Sigma_k^*$ and, using the usual direct product construction, creating
a DFA $M_5$ that accepts $L(M_3) \ \cap \ L(M_4)$.  Then $\bf a$ is 
ultimately periodic if and only if $M_5$ accepts some string, and this can
be checked using the usual depth-first search to look for a path connecting
the initial state with some final state.
\endpf
\end{proof}

\section{Decision problems about repetitions}
\label{squarefree}

     A morphism $h:\Sigma^* \rightarrow \Delta^*$ is said to be
$k$-power-free if whenever $w$ is $k$-power-free, so is $h(w)$.  
There is a reasonably large literature about these morphisms, with most
investigators concentrating on giving computable characterizations
of such morphisms; see, for example,
\cite{Berstel:1979,Crochemore:1982,Karhumaki:1983,Leconte:1985,Richomme&Wlazinski:2007}.

     We say a morphism $h:\Sigma^* \rightarrow \Sigma^*$
is {\it prolongable} on a letter $a$ if $h(a) = ax$ for some $x$ such that
$h^i(x) \not= \epsilon$ for all $i \geq 0$.    In this case
there is a unique infinite word with prefixes $h^i (a)$ for all $i \geq 0$,
which we write as $h^\omega (a)$.  Such a word is called {\it morphic}.
It is also of interest to
give computable characterizations of those $h$ for which $h^\omega (a)$
avoids various kind of repetitions.  (Note that it is possible for $h^\omega(a)$
to, for example, avoid squares, even if $h$ itself is not squarefree.
The morphism $g$ given above in Section~\ref{intro-sec} provides an
example.  Here ${\tt 212}$ is squarefree, but $g({\tt 212})$ is not.)

Berstel \cite{Berstel:1979} showed how
to decide if $h^\omega (a)$ is squarefree for three-letter alphabets.
Karhum\"aki \cite{Karhumaki:1983} showed how to decide if
$h^\omega(a)$ is overlap-free for two-letter alphabets.  
Later, Mignosi and S\'e\'ebold \cite{Mignosi&Seebold:1993}
gave a general algorithm for testing the $k$-power-freeness of 
$h^\omega(a)$ for arbitrary non-erasing
morphisms $h$ and integers $k \geq 2$.   Cassaigne \cite{Cassaigne:1994}
showed how to test if certain kinds of HD0L words avoid arbitrary patterns.

     The technique of Section~\ref{ultimate} can be modified to create
a decision procedure for the existence of many kinds of
repetitions in $k$-automatic sequences.    Our approach is both more and less
general than previous results in the literature.
It is less general because our technique
works only for uniform morphisms.  It is
more general because (a) it works not only for fixed points of uniform
morphisms, but also images of those fixed points (under a coding);
(b) it works for testing the $r$-power-freeness and $r^+$-power-freeness
of words, where $r$ is an arbitrary rational number $> 1$ -- a topic 
relatively 
unexplored in the literature until now
(but see \cite{Krieger:2009,Krieger:2007}); and (c) it works for
arbitrary alphabets.
We do not know how to make
our technique work for $r$ an irrational number.

The following theorem illustrates the technique.

\bigskip

\begin{theorem}
     The following question is decidable:  given a $k$-automatic
sequence ${\bf a} = (a_n)_{n \geq 0}$ represented by a DFAO, is
${\bf a}$ overlap-free?
\label{overl}
\end{theorem}

\begin{proof}
    The proof is very similar to the proof of Theorem~\ref{periodic}.
The sequence ${\bf a} = (a_n)_{n \geq 0}$ contains an overlap if and only
if there exist
integers $I \geq 0, T \geq 1$ such that $a_{I+J} = a_{I+T+J}$ for all
$J$, $0 \leq J \leq T$.  

      Given a DFAO $M = (Q, \Sigma, \Delta, \delta, q_0, \tau)$
for $\bf a$, we create an NFA $M_2$ that on input $(I,T)$ 
accepts if there exists an integer $J$, $0 \leq J \leq T$, such
that $a_{I+J} \not= a_{I+T+J}$.  To accomplish this, $M_2$ 
guesses the bits of $J$, verifies that $0 \leq J \leq T$, computes
$I+J$ and $I+T+J$ on the fly, and accepts if $a_{I+J} \not= a_{I+T+J}$.
As before, we handle the problem that the expansion of $I+T+J$ might be longer
than that of $I$ or $T$ by allowing inputs with leading zeroes (actually trailing,
since inputs are entered starting with the least significant digit).
To do so, we modify the accepting states of $M_2$ to get a new NFA
$M_3$, by making a state of $M_3$ accepting if it can be reached in $M_2$
from an accepting state along a path labeled $(0,0)^j$ for some $j \geq 0$.

We now convert $M_3$ to a DFA using the subset construction, and change
all accepting states to non-accepting and vice versa, obtaining a DFA
$M_4$.  Hence $M_4$ accepts if for all $J$ with $0 \leq J \leq T$ we have
$a_{I+J} = a_{I+T+J}$; i.e., there is an overlap of length $2T+1$ beginning
at position $I$ of $\bf a$.  Thus $\bf a$ contains an overlap if and only
if $M_4$ accepts $(I,T)$ for some integers $I \geq 0$ and $T \geq 1$,
which, as before, can be easily checked.

     Here are the full details for the construction of $M_2 = (Q', \Sigma_k \times \Sigma_k,
\delta', q'_0, F')$.  The states are $5$-tuples of the form
$[b,c,d,q,r]$ where $b$ is one of $<, =$, or $>$, expressing the relationship
between the guessed $J$ and the input $T$; $c$ is the carry in the computation of
$I+J$; $d$ is the carry in the computation of $I+T+J$; $q$ is the state of
$M$ reached on input $I+J$; and $r$ is the state of $M$ reached on input
$I+T+J$.   The initial state is $q'_0 = [=, 0, 0, q_0, q_0]$, and the
set of final states is
$$F' = \lbrace [ b, 0, 0, q, r] \ : \ b \in \lbrace <, = \rbrace
\text{ and } \tau(q) \not= \tau(r) \rbrace .$$
Finally, $\delta'$ is defined as follows:
$$ \delta'([b,c,d,q,r], (i', t')) = 
\lbrace \ [ u(b,j', t'), \lfloor {{c + i' + j'} \over k} \rfloor, $$
$$
 \lfloor {{d + i' + j' + t'} \over k} \rfloor,
  \delta(q, (c + i' + j') \bmod k), \delta(r, (d + i' + j' + t') \bmod k)] \ : \ 0 \leq j' < k \rbrace.$$
  \endpf
\end{proof}

\begin{example}
Using the Grail package \cite{Raymond&Wood:1994}, version 3.3.4,
we verified purely mechanically that the Thue-Morse word $\bf t$
is overlap-free.  We carried out the construction of Theorem~\ref{overl}
by creating an NFA of $72$ states ($3$ possibilities for $b$,
$2$ for $c$, $3$ for $d$ (since carries for $d + i' + j' + t'$ could be as much as
$2$), and 2 possibilities for each of $q$ and $r$).  We added the correct final
states, and then converted this to a DFA with 801 states.  We then took the complement
of this DFA, obtaining a DFA that accepts all pairs
$(I,T)$ where there is an overlap of length $2T+1$ beginning at position
$I$.
We then minimized, obtaining
a DFA with $2$ states that only accepts strings corresponding to $T = 0$.
Hence $\bf t$ is overlap-free.
\end{example}


%

      The same idea can be used to prove each of the following results:

\bigskip

\begin{theorem}
     Given a DFAO $M$ generating a $k$-automatic sequence $\bf a$,
each of the following properties is decidable:

\begin{itemize}

\item[(a)] Given a rational number $r$, whether
$\bf a$ avoids $r$-powers (resp., $r^+$-powers);

\item[(b)] Given a rational number $r$, whether
$\bf a$ contains infinitely many occurrences of $r$-powers (resp.,
$r^+$-powers);

\item[(c)] Given a rational number $r$, whether
$\bf a$ contains infinitely many distinct $r$-powers (resp.,
$r^+$-powers);

\item[(d)] Given a rational number $r$, and a length $l$, whether
$\bf a$ avoids $x^r$ (resp., $r^+$-powers) for $|x| \geq l$;

\item[(e)] Given a rational number $r$, whether $\bf a$ avoids
$x^r$ for all sufficiently long $x$;

\item[(f)] Given a length $l$, whether $\bf a$ avoids palindromes of
length $\geq l$ 
(cf.\ \cite{Rampersad&Shallit:2005});

\item[(g)] Whether $\bf a$ avoids all sufficiently long palindromes;

\item[(h)] Given a length $l$, whether $\bf a$ satisfies the property that
$x$ is a factor of $\bf a$ of length $\geq l$, then its reverse $x^R$ is not
(cf.\ \cite{Rampersad&Shallit:2005});

\item[(i)] Assuming $\bf a$ is defined over the alphabet
$\lbrace 0, 1, \ldots, j-1 \rbrace$, whether
$\bf a$ avoids all factors of the form
$x \sigma(x)$ where $\sigma(a) = (a+1) \bmod j$ (cf.\ \cite{Loftus&Shallit&Wang:2000}).

\end{itemize}
\end{theorem}

     The proofs for each part are more-or-less trivial variations
on the proof of Theorem~\ref{overl}, and we omit them.  However, we do
make one remark:
for parts (a)-(e), we need to replace the condition for the existence
of overlaps, namely,
``there exist $I \geq 0, T \geq 1$ such that $a_{I+J} = a_{I+T+J}$ for all
$J$, $0 \leq J \leq T$'' with the appropriate condition for $\alpha$-powers,
where $\alpha = {p \over q}$ is a rational number.  The new condition is
``there exist $I \geq 0, T \geq 1$ such that $a_{I+J} = a_{I+T+J}$ for all
$J$, $0 \leq J < ({p \over q} - 1)T$''.  (In the case of $\alpha^+$-powers,
the inequality becomes $0 \leq J \leq ({p \over q} - 1)T$.)
At first sight it might seem difficult to implement this test,
for although multiplication can be carried out easily starting with the
least significant digit, division is more problematic.  To handle this,
we simply rewrite the inequality $J < ({p \over q} -1)T$ as
$qJ < (p-q)T$.  Now on input $T$ we can guess $J$ digit-by-digit, transduce
$J$ into $qJ$ and $T$ into $(p-q)T$, and verify
the inequality $qJ < (p-q)T$ on the fly starting with the least
significant digit, as before.  

\section{The orbit closure}
\label{orbit}

     We now turn to orbits and the orbit closure of automatic sequences.
As motivation, recall that a certain classical dynamical system (i.e.,
a compact set together with a continuous map of this set) is associated with
any sequence, namely the topological closure of the orbit of that sequence under
the shift. For
some sequences, the lexicographically least and largest sequences in
the orbit closure are known explicitly.

     Consider, as an example, the Thue-Morse sequence $\bf t$.  
The lexicographically least sequence in the orbit closure
of $\bf t$ is the sequence obtained by iterating the Thue-Morse morphism
$\mu: {\tt 0} \rightarrow {\tt 01}, {\tt 1} \rightarrow {\tt 10}$ on
$\tt 1$, and then dropping the first letter
\cite{Allouche:1983,Allouche&Cosnard:1983,Allouche&Currie&Shallit:1998,Komornik&Loreti:1998}.
This gives
$$ {\tt 001011001101001} \cdots $$
and this sequence is clearly $2$-automatic, as it is accepted by the
DFAO in Figure~\ref{llt} below.

\begin{figure}[H]
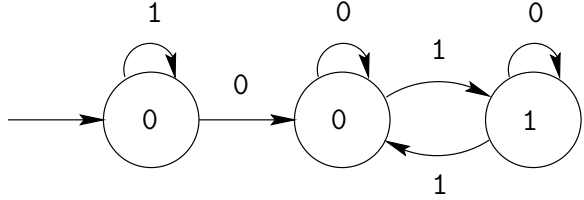

\begin{center}
\input orb5.tex
\end{center}
\caption{Automaton generating the lexicographically least sequence in the orbit
closure of the Thue-Morse sequence}
\label{llt}
\end{figure}

Other examples are discussed in Section~\ref{appli-sec}.  Recall
that the Rudin-Shapiro sequence 
${\bf u} = (u_n)_{n \geq 0}$ is a $2$-automatic sequence defined
as follows:  $u_n$ is $0$ or $1$ according to whether the number of
(possibly overlapping) occurrences of {\tt 11} in the binary expansion
of $n$ is even or odd.   We observe empirically that the lexicographically least
sequence in the orbit closure of the Rudin-Shapiro sequence 
seems to be the sequence obtained by preceding the Rudin-Shapiro
sequence by a $0$, but we did not yet prove this.


     We now apply the technique of Section~\ref{ultimate}
to the lexicographically least sequence
in the orbit closure of a $k$-automatic sequence.   
Our idea is based on the following
characterization.

\bigskip

\begin{lemma}
Let
${\bf a} = (a_n)_{n \geq 0}$ be a sequence, and let
${\bf b} = (b_n)_{n \geq 0}$ be the lexicographically least sequence in
the orbit closure of $\bf a$. Then
$b_i = c$ if and only if there exists $j \geq 0$ such that
$a_{j+i} = c$ and 
$a_l a_{l+1} \cdots a_{l+i} \geq a_j a_{j+1} \cdots a_{j+i}$
for all $l \geq 0$.  
\label{char}
\end{lemma}

\begin{proof}
Suppose $b_i = c$.  Then there exists $j \geq 0$ such that
$a_j a_{j+1} \cdots a_{j+i} = b_0 b_1 \cdots b_i$, so $a_{j+i} = b_i$.
But then 
$a_l a_{l+1} \cdots a_{l+i} \geq a_j a_{j+1} \cdots a_{j+i}$
for all $l \geq 0$.   (Here we use $\geq$ for lexicographic order.)

On the other hand, if $a_l a_{l+1} \cdots a_{l+i} \geq a_j a_{j+1}
\cdots a_{j+i}$ for all $l \geq 0$, then $a_j a_{j+1}
\cdots a_{j+i}$ must be the prefix of $\bf b$ of length $i+1$, and
so $b_i = a_{j+i} = c$.
\endpf
\end{proof}

The advantage to this characterization of $b_i$ is that it does not
require explicit knowledge of $b_0, b_1, \ldots, b_{i-1}$.

\bigskip

\begin{theorem}
Let $\bf a$ be $k$-automatic, and
let ${\bf b}$ be the lexicographically least sequence in the orbit closure
of ${\bf a}$.   Then $\bf b$ is $k$-automatic.
\label{main}
\end{theorem}

\begin{proof}
The idea is to use the condition in Lemma~\ref{char}.
The proof is similar to the proof of Theorem~\ref{periodic}, and we outline
it below.  The fine details about how everything is computed are similar
to those of Theorem~\ref{periodic} and we omit them.

The proof consists of several steps.  First, suppose we have a
$k$-DFAO $M$ generating $\bf a$.  We now create an NFA $M_1$ that
on input $(L, J, R)$ accepts if and only if
there exists $t$, $0 \leq t < R$,
such that $a_{L+t} \not= a_{J+t}$, or $a_{L+R} \geq a_{J+R}$.  
The idea is to ``guess'' $t$ bit-by-bit, verify the inequality
$0 \leq t < R$, while simultaneously computing
the quantities $L+t$, $J+t$, $L+R$, and $J+R$.  We accept if
$a_{L+t} \not= a_{J+t}$ for some $t$, $0 \leq t < R$, or if
$a_{L+R} \geq a_{J+R}$.

From $M_1$ we create a DFA $M_2$ that on input $(L,J,R)$ accepts
if and only if $a_{L+t} = a_{J+t}$ for all $t$, $0 \leq t < R$ and
$a_{L+R} < a_{J+R}$.  This is done by converting $M_1$ to a DFA using
the subset construction and changing all accepting states to non-accepting
and vice versa.  Thus $M_2$ accepts $(L,J,R)$ if and only if
$a_L a_{L+1} \cdots a_{L+R} < a_J a_{J+1} \cdots a_{J+R}$.  

Next, from $M_2$ we create an NFA $M_3$ that on input $(J,R)$ accepts
if and only if there exists an $L \geq 0$ such that
$a_L a_{L+1} \cdots a_{L+R} < a_J a_{J+1} \cdots a_{J+R}$.  The idea
is to ``guess'' $L$ bit-by-bit and call $M_2$ on $(L,J,R)$.  A priori
$L$ could be very big compared to $J$ and $R$, but our previous trick to
handle this works.  

Then from $M_3$ we create a DFA $M_4$ that on input $(J,R)$ accepts if and
only if for all $L \geq 0$ we have
$a_L a_{L+1} \cdots a_{L+R} \geq a_J a_{J+1} \cdots a_{J+R}$.  This
is done by converting $M_3$ to a DFA using the subset construction, and
then changing all accepting states to non-accepting
and vice versa. 

From $M_4$ we create an NFA $M_5$ that on input $cI$ (i.e., the character
$c$ concatenated with the base-$k$ expansion of $I$) accepts if and only
if there exists $J \geq 0$ with $a_{J+I} = c$ and
$a_L a_{L+1} \cdots a_{L+I} \geq a_J a_{J+1} \cdots a_{J+I}$ for
all $L \geq 0$.  This is done by recording $c$ in the state,
``guessing'' $J$ bit-by-bit, computing $J+I$ bit-by-bit and simulating
$M$ on $J+I$, and
calling $M_4$ with input $(J,I)$.   We then convert $M_5$ to a
DFA $M_6$ using the subset construction.

Finally, we create a $k$-DFAO $M_7$ that on input $I$ simulates
$M_6$ on input $cI$ in parallel for each $c \in \Delta$.  Exactly one
branch will accept, and the output associated with this branch is $c$.
\endpf
\end{proof}

\section{Continued fraction expansions}

     The results of the previous section
can be generalized to other kinds of orders.
Instead of the ordinary lexicographic order, we could consider an order
that depends on the index of the string being compared.  One way to do
this is to consider a sequence of permutations $(\psi_i)_{i \geq 0}$, where
each $\psi_i: \Delta \rightarrow \Delta$, and
when comparing $a_0 a_1 \cdots a_{i-1}$ to $b_0 b_1 \cdots b_{i-1}$,
we instead compare 
$\psi_0 (a_0) \cdots \psi_{i-1} (a_{i-1})$ to
$\psi_0 (b_0) \cdots \psi_{i-1} (b_{i-1})$ (using the ordinary lexicographic
order).  
An example of 
this kind of ordering comes from continued fractions, where 
$[a_0, a_1, a_2, \ldots ] < [b_0, b_1, b_2, \ldots]$ if and only
if $a_0 < b_0$, or $a_0 = b_0$ and $a_1 > b_1$, or
$a_0 = b_0$, $a_1 = b_1$, and $a_2 < b_2$, etc.  This corresponds
to inverting the order of the elements being compared on the odd indexes.
Provided the sequence $(\psi_i)_{i \geq 0}$ is $k$-automatic, 
the result of Theorem~\ref{main} still holds.

\bigskip

\begin{corollary}
      Let $(\psi_i)_{i \geq 0}$ be a $k$-automatic sequence of
permutations, and let $(a_i)_{i \geq 0}$ be a $k$-automatic sequence.
Then the lexicographically least sequence in the orbit closure,
as modified by the permutations $(\psi_i)$, is $k$-automatic.
\label{perm}
\end{corollary}

\begin{proof}
      In the construction of Theorem~\ref{main}, when we compare
$a_{L+R}$ to $a_{J+R}$, we instead compare
$\psi_R (a_{L+R})$ to $\psi_R(a_{J+R})$.  Since $(\psi_i)_{i \geq 0}$
is $k$-automatic, there is no problem computing $\psi_R$ on input $R$.
\endpf
\end{proof}

     From now on, when we talk about a continued fraction expansion $[a_0, a_1,
\ldots]$ being $k$-automatic, we mean the continued fraction has bounded
partial quotients and the underlying
sequence of partial quotients $(a_i)_{i \geq 0}$
is $k$-automatic.

     Let $T(x)$ be the usual transformation on continued fractions
defined by $T(x) = {1 \over {x - \lfloor x \rfloor}}$, so that
$T([a_0, a_1, a_2, \ldots ] ) = [a_1, a_2, \ldots ]$.  Thus we have

\bigskip

\begin{theorem}
      Let $x$ be an irrational real number with a
$k$-automatic continued fraction
expansion $[a_0, a_1, \ldots]$.
Then the continued fraction expansions of both
$\liminf_{n \rightarrow \infty} T^n (x)$
and $\limsup_{n \rightarrow \infty} T^n (x)$ are 
$k$-automatic.
\end{theorem}

\begin{proof}
     Use Corollary~\ref{perm}, where the permutations invert the order
of the letters on every other index.
\endpf
\end{proof}

     In addition to the orbit closure of a sequence, we can study a related
structure, which we call the {\it reverse orbit closure}.  
We say that a sequence ${\bf b} = (b_n)_{n \geq 0}$ is in the
reverse orbit closure of ${\bf a} = (a_n)_{n \geq 0}$ if every finite
prefix of $\bf b$ is a prefix of some word of the form
$a_r a_{r-1} a_{r-2} \cdots a_1 a_0$.  

\bigskip

\begin{theorem}
If ${\bf a} = (a_n)_{n \geq 0}$ is $k$-automatic, then so is the 
lexicographically least sequence in the reverse orbit closure.
\end{theorem}

\begin{proof}
Let $b = (b_n)_{n \geq 0}$ be the lexicographically 
least sequence
in the reverse orbit closure of ${\bf a} = (a_n)_{n \geq 0}$.  We use
the following characterization of $\bf b$:
$b_i = c$ if and only if
there exists  $r \geq i$ 
such that  $a_{r-i} = c$ and
$a_s a_{s-1} \cdots a_{s-i} \geq a_r a_{r-1} \cdots a_{r-i} $
 for all $ s \geq i $.

We can now implement this test in exactly the same way that we
implemented the test in the proof of Theorem~\ref{main}.
\endpf
\end{proof}

     We can also combine the reverse orbit closure with
a permutation that inverts the order of the letters on every other
index.  

\begin{theorem}
Let $\alpha$ be an irrational
real number with a $k$-automatic continued fraction expansion
$[a_0, a_1, a_2, \ldots]$.
Let $p_n/q_n$ be the $n$'th convergent to the
continued fraction to $\alpha$.  Let $\beta = \liminf_{n \rightarrow \infty} p_n/p_{n-1}$
and $\gamma = \liminf_{n \rightarrow \infty} q_n/q_{n-1}$, $\delta = \limsup_{n \rightarrow \infty} p_n/p_{n-1}$,
$\zeta = \limsup_{n \rightarrow \infty} q_n/q_{n-1}$.  
Then the continued fraction expansion of each of
$\beta, \gamma, \delta, \zeta$ is $k$-automatic.
\end{theorem}

\begin{proof}
We prove the result for $\beta$, the others being similar.
By a famous result of Galois \cite{Galois:1828} we have
$${{p_n} \over {p_{n-1}}} = [a_n, a_{n-1}, \ldots, a_0 ].$$
Now $\beta$ corresponds to the lexicographically least sequence in
the reverse orbit closure of $(a_i)_{i \geq 0}$,
except that the ordering is slightly
different from the usual ordering, where the ordering is as usual on the
even indexed terms and opposite on the odd-indexed terms.  As in
Corollary~\ref{perm}, we can handle this in the same
way.
\endpf
\end{proof}

\noindent{\bf Example.}
     Let us consider an example.  As is well-known \cite{Shallit:1979,vanderPoorten&Shallit:1992}, for integers $k \geq 3$ the real number
$$\alpha_k = \sum_{i \geq 0} k^{-2^i} = [0,k-1,k+2,k,k,k-2,k,k+2,k,k-2,k+2,k,k-1,\ldots]$$
has a $2$-automatic continued fraction expansion, generated by the automaton
given in Figure~\ref{alphak} (again, the automaton
expects the least significant digit first).

\begin{figure}[H]
\begin{center}
\input orb8.tex
\end{center}
\caption{Automaton generating the continued fraction for $\alpha_k$}
\label{alphak}
\end{figure}

\begin{figure}[H]
\begin{center}
\input orb9.tex
\end{center}
\caption{Automaton generating the continued fraction for $\zeta_k$}
\label{betak}
\end{figure}

Then $\zeta_k = \limsup_{n \geq 0} q_n/q_{n-1} = [k+2, k-2, k , k+2,
k, k-2, k, k, \ldots ]$ is $2$-automatic.   \endpf

\bigskip

Let $\alpha$ be an irrational number with partial quotients $p_n/q_n$.
The quantity $\zeta = \limsup_{n \geq 0} q_n/q_{n-1}$ figures in a number of
recent papers in combinatorics on words.  For example, $2+\zeta$ is the
value of the recurrence quotient of a Sturmian word with slope
$\alpha$ \cite{Cassaigne:1999,Adamczewski&Allouche:2007}.  Hence
this recurrence quotient has a $k$-automatic continued fraction if
$\alpha$ does.

The number $\zeta$ also appears (actually, $\zeta+1$) as
the irrationality measure of numbers of the form
$(b-1) \sum_{n \geq 1} b^{-\lfloor n \alpha \rfloor}$ 
\cite{Adamczewski&Allouche:2007}.

Finally, $\zeta$  also appears in a formula giving the critical exponent (aka
``index'') of
Sturmian words, as found by 
Damanik and Lenz \cite[Thm.\ 1, p.\ 24]{Damanik&Lenz:2002}
and Cao and Wen \cite[Thm.\ 9, p.\ 380]{Cao&Wen:2003}.  This exponent
is essentially 
$$\zeta' :=  2 + \limsup_{n \geq 1} {{q_n - 2} \over { q_{n-1}}}.$$
If the $\limsup$ is actually attained for a particular $n$, then the
critical exponent is rational.  Otherwise it clearly coincides
with $2 + \zeta$, and its continued fraction expansion is
$k$-automatic if that of $\alpha$ is.

\section{Applications}
\label{appli-sec}

Our results about the lexicographically least and largest sequences 
in the orbit closure of a sequence can be illustrated by and applied to
two families of binary sequences: the sequences in the set $\Gamma$
described below and the Sturmian sequences.
 
\subsection{Sequences in the set $\Gamma$}

Theorem~\ref{main} can be applied to shed some light on the
automatic sequences that belong to two sets of binary sequences:
the set $\Gamma$ occurring in the study of iterations of continuous 
unimodal maps of the interval (see \cite{Allouche&Cosnard:1983,Allouche:1983})
and the set $\Gamma_{\rm strict}$ occurring in the study of unique
$\beta$-expansions of the number $1$
\cite{Erdos&Joo&Komornik:1990,Komornik&Loreti:1998,Allouche&Cosnard:2001},
where
$$
\begin{array}{lll}
\Gamma &:=& \{A \in \{ {\tt 0, 1}\}^\omega:  \ \forall k \geq 0, \
\overline{A} \leq \sigma^k A \leq A\} \\
& & \\
\Gamma_{\rm strict} &:=& \{A \in \{{\tt 0, 1}\}^\omega: \ \forall k \geq 1,
\ \overline{A} < \sigma^k A < A\}. \\
\end{array}
$$
Here $A = (a_n)_{n \geq 0}$, and
$\sigma$ is the shift on sequences defined by  $\sigma A :=
(a_{n+1})_{n \geq 0}$. The bar operation replaces $0$'s by $1$'s
and $1$'s by $0$'s, i.e., $\overline{A} := (1-a_n)_{n \geq 0}$.
Note that these two sets differ only by a set of (purely) periodic sequences.
Also note that the set $\Gamma$ above differs slightly from the set $\Gamma$
in \cite{Allouche:1983}, in that
the set $\Gamma$ above contains the extra sequence 
$({\tt 10})^{\omega}$.

The shifted Thue-Morse sequence is an element of $\Gamma$, as are more 
general automatic sequences (e.g., analogues of the Thue-Morse sequence 
including the $q$-mirror sequences introduced in \cite{Allouche&Cosnard:1983,Allouche:1983}; see
\cite{Komornik&Loreti:2002,Komornik&Loreti:2007,deVries&Komornik:2007,Niu&Wen:2006,Allouche&Frougny:2007}).

Now for any binary sequence $A$ belonging to $\Gamma$, define, as in
\cite{Allouche&Cosnard:1983,Allouche:1983},
$$
\Gamma_A := \{B \in \{{\tt 0, 1}\}^\omega: \ \forall k \geq 0, \ 
\overline{A} \leq \sigma^k B \leq A \}. 
$$
Of course, the sequence $A$ belongs to $\Gamma_A$. Furthermore ${\tt 1}^{\omega}$ 
belongs to $\Gamma$, and any binary sequence $B$ belongs to 
$\Gamma_{{\tt 1}^{\omega}}$. Thus, given $B$, it is interesting to look for the 
lexicographically
least sequence $A$ such that $B$ belongs to $\Gamma_A$. The answer is easy 
(see \cite[pp.~37--38]{Allouche:1983}):
the least sequence $A$ in $\Gamma$ such that $B$ belongs to $\Gamma_A$
is 
$$
\Theta(B) := 
\sup(\{\sigma^k B: \ k \geq 0\} \cup
\{\sigma^{\ell} \overline{B}: \ \ell \geq 0\}).
$$
In particular for any sequence $B$, all sequences $\sigma^k B$ and
$\sigma^{\ell} \overline{B}$ belong to $\Gamma_{\Theta(B)}$, and $\Theta(B)$
is the largest such sequence.

Theorem~\ref{main} above shows that if $B$ is automatic, then so is
$\Theta(B)$. This remark is a small step in the study of {\em all\,} automatic 
sequences belonging to $\Gamma$. Note that $\Gamma$ is not countable (see 
e.g., \cite[Prop.~3, p.~35]{Allouche:1983}), so that $\Gamma$ also
contains sequences 
that are {\em not\,} automatic. Even more, $\Gamma$ contains sequences whose
subword complexity is not $O(n)$: it suffices to take the sequence $\Theta(B)$,
where $B$ is, as in \cite{Grillenberger:1973}, a binary minimal sequence with positive 
topological entropy, hence with subword complexity not of the form $O(n)$.

\subsection{Sturmian sequences}

We suppose that the reader is familiar with the notion of Sturmian
sequence (see, e.g., \cite[Chapter~2]{Lothaire:2002}).
A result on characteristic 
Sturmian sequences and Sturmian sequences that was proved or partly 
proved several times (see the survey \cite{Allouche&Glen:2008}) states that

\bigskip

\begin{theorem}\label{extr}

\ { }

\begin{itemize}

\item[(a)] 
A nonperiodic sequence $A$ is characteristic Sturmian if 
and only if for any $k \geq 0$ the following inequalities hold
$$
{\tt 0} \, A \leq \sigma^k A \leq {\tt 1} \, A.
$$

\item[(b)]
A nonperiodic binary sequence $A$ is Sturmian if and only if there exists 
a binary sequence $B$ such that for any $k \geq 0$ the following inequalities 
hold
$$
{\tt 0} \, B \leq \sigma^k A \leq {\tt 1} \, B.
$$
Furthermore such a $B$ is unique, and is the characteristic Sturmian sequence
having the same slope as $A$.

\end{itemize}

\end{theorem}

Theorem~\ref{extr} easily implies the following corollary.

\bigskip

\begin{corollary}
The lexicographically least (resp.~largest) sequence in the orbit closure of 
a Sturmian sequence $A$ is the sequence ${\tt 0}B$ (resp.~${\tt 1}B$)
where $B$ is the
characteristic sequence with the same slope as $A$.
\end{corollary}

\begin{proof}
It is not difficult to see that the inequalities above are optimal in the
sense that, e.g., for a characteristic sequence $A$, we have
${\tt 0} A = \inf\{\sigma^k A \ :  \ k \geq 0\}$
and similarly for the other three
inequalities in Theorem~\ref{extr} above.  
\endpf
\end{proof}

\section{Acknowledgments}

We thank Kalle Saari for his help with Finnish and for pointing out an
error in a previous version.  We thank the referees for a careful reading
of the paper.

\end{document}

%% file: orb2.tex
\begin{picture}(0,0)%
\epsfig{file=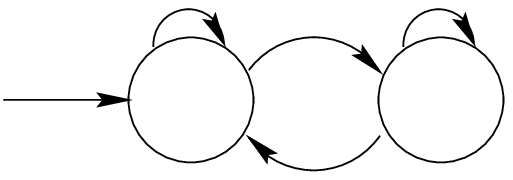}%
\end{picture}%
\setlength{\unitlength}{3947sp}%
\begingroup\makeatletter\ifx\SetFigFont\undefined%
\gdef\SetFigFont#1#2#3#4#5{%
  \reset@font\fontsize{#1}{#2pt}%
  \fontfamily{#3}\fontseries{#4}\fontshape{#5}%
  \selectfont}%
\fi\endgroup%
\begin{picture}(2420,1260)(1039,-766)
\put(1876,321){\makebox(0,0)[lb]{\smash{\SetFigFont{12}{14.4}{\rmdefault}{\mddefault}{\updefault}{\tt 0}}}}
\put(3046,329){\makebox(0,0)[lb]{\smash{\SetFigFont{12}{14.4}{\rmdefault}{\mddefault}{\updefault}{\tt 0}}}}
\put(2491,209){\makebox(0,0)[lb]{\smash{\SetFigFont{12}{14.4}{\rmdefault}{\mddefault}{\updefault}{\tt 1}}}}
\put(2491,-766){\makebox(0,0)[lb]{\smash{\SetFigFont{12}{14.4}{\rmdefault}{\mddefault}{\updefault}{\tt 1}}}}
\put(1891,-256){\makebox(0,0)[lb]{\smash{\SetFigFont{12}{14.4}{\rmdefault}{\mddefault}{\updefault}{\tt 0}}}}
\put(3083,-264){\makebox(0,0)[lb]{\smash{\SetFigFont{12}{14.4}{\rmdefault}{\mddefault}{\updefault}{\tt 1}}}}
\end{picture}

%% file: orb1.tex
\begin{picture}(0,0)%
\epsfig{file=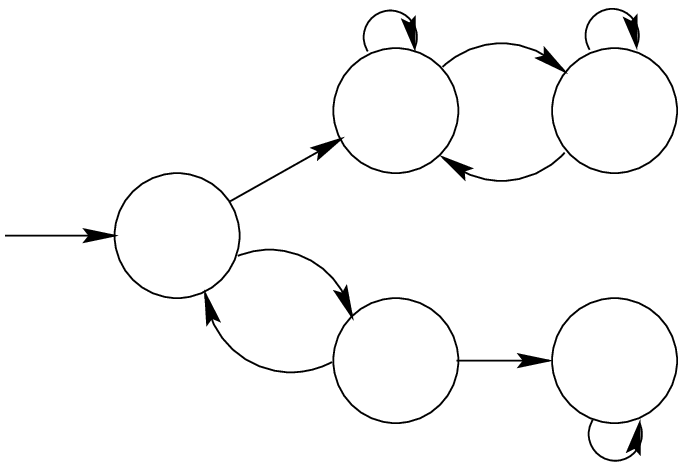}%
\end{picture}%
\setlength{\unitlength}{3947sp}%
\begingroup\makeatletter\ifx\SetFigFont\undefined%
\gdef\SetFigFont#1#2#3#4#5{%
  \reset@font\fontsize{#1}{#2pt}%
  \fontfamily{#3}\fontseries{#4}\fontshape{#5}%
  \selectfont}%
\fi\endgroup%
\begin{picture}(3245,2654)(664,-2243)
\put(1831,-548){\makebox(0,0)[lb]{\smash{\SetFigFont{12}{14.4}{\rmdefault}{\mddefault}{\updefault}{\tt 0}}}}
\put(1740,-1817){\makebox(0,0)[lb]{\smash{\SetFigFont{12}{14.4}{\rmdefault}{\mddefault}{\updefault}{\tt 1}}}}
\put(1433,-1013){\makebox(0,0)[lb]{\smash{\SetFigFont{12}{14.4}{\rmdefault}{\mddefault}{\updefault}{\tt 2}}}}
\put(2483,-1621){\makebox(0,0)[lb]{\smash{\SetFigFont{12}{14.4}{\rmdefault}{\mddefault}{\updefault}{\tt 1}}}}
\put(3009,-1419){\makebox(0,0)[lb]{\smash{\SetFigFont{12}{14.4}{\rmdefault}{\mddefault}{\updefault}{\tt 0}}}}
\put(3458,-2243){\makebox(0,0)[lb]{\smash{\SetFigFont{12}{14.4}{\rmdefault}{\mddefault}{\updefault}{\tt 0,1}}}}
\put(3015, 67){\makebox(0,0)[lb]{\smash{\SetFigFont{12}{14.4}{\rmdefault}{\mddefault}{\updefault}{\tt 1}}}}
\put(3533,-413){\makebox(0,0)[lb]{\smash{\SetFigFont{12}{14.4}{\rmdefault}{\mddefault}{\updefault}{\tt 0}}}}
\put(2484,-405){\makebox(0,0)[lb]{\smash{\SetFigFont{12}{14.4}{\rmdefault}{\mddefault}{\updefault}{\tt 2}}}}
\put(2461,232){\makebox(0,0)[lb]{\smash{\SetFigFont{12}{14.4}{\rmdefault}{\mddefault}{\updefault}{\tt 0}}}}
\put(3016,-917){\makebox(0,0)[lb]{\smash{\SetFigFont{12}{14.4}{\rmdefault}{\mddefault}{\updefault}{\tt 1}}}}
\put(2197,-1006){\makebox(0,0)[lb]{\smash{\SetFigFont{12}{14.4}{\rmdefault}{\mddefault}{\updefault}{\tt 1}}}}
\put(3532,-1606){\makebox(0,0)[lb]{\smash{\SetFigFont{12}{14.4}{\rmdefault}{\mddefault}{\updefault}{\tt 1}}}}
\put(3557,246){\makebox(0,0)[lb]{\smash{\SetFigFont{12}{14.4}{\rmdefault}{\mddefault}{\updefault}{\tt 0}}}}
\end{picture}

%% file: orb5.tex
\begin{picture}(0,0)%
\epsfig{file=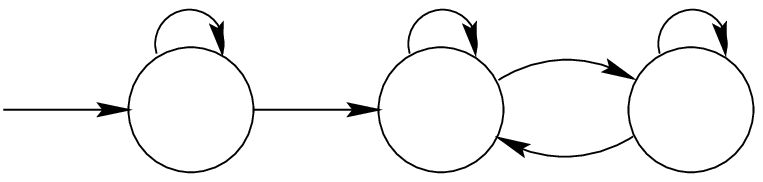}%
\end{picture}%
\setlength{\unitlength}{3947sp}%
\begingroup\makeatletter\ifx\SetFigFont\undefined%
\gdef\SetFigFont#1#2#3#4#5{%
  \reset@font\fontsize{#1}{#2pt}%
  \fontfamily{#3}\fontseries{#4}\fontshape{#5}%
  \selectfont}%
\fi\endgroup%
\begin{picture}(3620,1253)(589,-668)
\put(3270,-668){\makebox(0,0)[lb]{\smash{\SetFigFont{12}{14.4}{\rmdefault}{\mddefault}{\updefault}{\tt 1}}}}
\put(3263,179){\makebox(0,0)[lb]{\smash{\SetFigFont{12}{14.4}{\rmdefault}{\mddefault}{\updefault}{\tt 1}}}}
\put(2663,420){\makebox(0,0)[lb]{\smash{\SetFigFont{12}{14.4}{\rmdefault}{\mddefault}{\updefault}{\tt 0}}}}
\put(3871,412){\makebox(0,0)[lb]{\smash{\SetFigFont{12}{14.4}{\rmdefault}{\mddefault}{\updefault}{\tt 0}}}}
\put(3833,-264){\makebox(0,0)[lb]{\smash{\SetFigFont{12}{14.4}{\rmdefault}{\mddefault}{\updefault}{\tt 1}}}}
\put(2633,-257){\makebox(0,0)[lb]{\smash{\SetFigFont{12}{14.4}{\rmdefault}{\mddefault}{\updefault}{\tt 0}}}}
\put(1448,-256){\makebox(0,0)[lb]{\smash{\SetFigFont{12}{14.4}{\rmdefault}{\mddefault}{\updefault}{\tt 0}}}}
\put(2019,-38){\makebox(0,0)[lb]{\smash{\SetFigFont{12}{14.4}{\rmdefault}{\mddefault}{\updefault}{\tt 0}}}}
\put(1478,412){\makebox(0,0)[lb]{\smash{\SetFigFont{12}{14.4}{\rmdefault}{\mddefault}{\updefault}{\tt 1}}}}
\end{picture}

%% file: orb8.tex
\begin{picture}(0,0)%
\epsfig{file=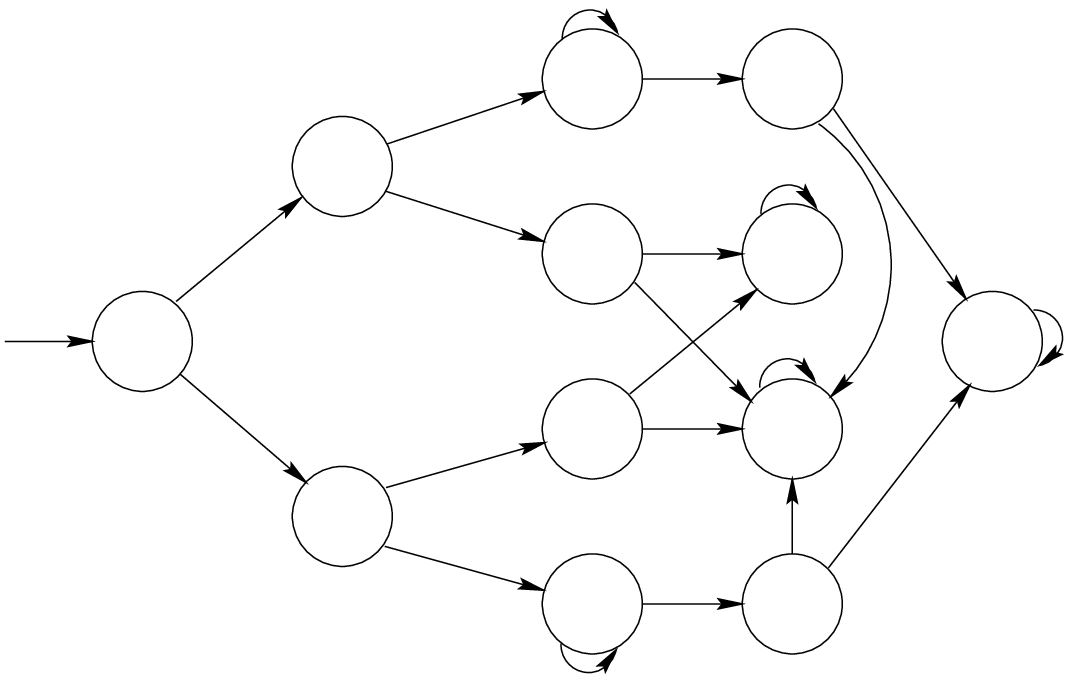}%
\end{picture}%
\setlength{\unitlength}{3158sp}%
\begingroup\makeatletter\ifx\SetFigFont\undefined%
\gdef\SetFigFont#1#2#3#4#5{%
  \reset@font\fontsize{#1}{#2pt}%
  \fontfamily{#3}\fontseries{#4}\fontshape{#5}%
  \selectfont}%
\fi\endgroup%
\begin{picture}(6499,4546)(664,-4322)
\put(3219,-1538){\makebox(0,0)[lb]{\smash{\SetFigFont{10}{12.0}{\rmdefault}{\mddefault}{\updefault}{\tt 1}}}}
\put(3241,-616){\makebox(0,0)[lb]{\smash{\SetFigFont{10}{12.0}{\rmdefault}{\mddefault}{\updefault}{\tt 0}}}}
\put(2626,-1081){\makebox(0,0)[lb]{\smash{\SetFigFont{10}{12.0}{\rmdefault}{\mddefault}{\updefault}$0$}}}
\put(2484,-3188){\makebox(0,0)[lb]{\smash{\SetFigFont{10}{12.0}{\rmdefault}{\mddefault}{\updefault}$k-1$}}}
\put(1824,-2730){\makebox(0,0)[lb]{\smash{\SetFigFont{10}{12.0}{\rmdefault}{\mddefault}{\updefault}{\tt 1}}}}
\put(1876,-1411){\makebox(0,0)[lb]{\smash{\SetFigFont{10}{12.0}{\rmdefault}{\mddefault}{\updefault}{\tt 0}}}}
\put(4615,-346){\makebox(0,0)[lb]{\smash{\SetFigFont{10}{12.0}{\rmdefault}{\mddefault}{\updefault}{\tt 1}}}}
\put(4375,-1981){\makebox(0,0)[lb]{\smash{\SetFigFont{10}{12.0}{\rmdefault}{\mddefault}{\updefault}{\tt 1}}}}
\put(4705,-3526){\makebox(0,0)[lb]{\smash{\SetFigFont{10}{12.0}{\rmdefault}{\mddefault}{\updefault}{\tt 1}}}}
\put(4112,-4322){\makebox(0,0)[lb]{\smash{\SetFigFont{10}{12.0}{\rmdefault}{\mddefault}{\updefault}{\tt 0}}}}
\put(4674,-1404){\makebox(0,0)[lb]{\smash{\SetFigFont{10}{12.0}{\rmdefault}{\mddefault}{\updefault}{\tt 0}}}}
\put(4097, 59){\makebox(0,0)[lb]{\smash{\SetFigFont{10}{12.0}{\rmdefault}{\mddefault}{\updefault}{\tt 0}}}}
\put(5319,-2656){\makebox(0,0)[lb]{\smash{\SetFigFont{10}{12.0}{\rmdefault}{\mddefault}{\updefault}$k$}}}
\put(5319,-556){\makebox(0,0)[lb]{\smash{\SetFigFont{10}{12.0}{\rmdefault}{\mddefault}{\updefault}$k$}}}
\put(4141,-548){\makebox(0,0)[lb]{\smash{\SetFigFont{10}{12.0}{\rmdefault}{\mddefault}{\updefault}$0$}}}
\put(1441,-2123){\makebox(0,0)[lb]{\smash{\SetFigFont{10}{12.0}{\rmdefault}{\mddefault}{\updefault}$0$}}}
\put(3272,-2670){\makebox(0,0)[lb]{\smash{\SetFigFont{10}{12.0}{\rmdefault}{\mddefault}{\updefault}{\tt 1}}}}
\put(5094,-1074){\makebox(0,0)[lb]{\smash{\SetFigFont{10}{12.0}{\rmdefault}{\mddefault}{\updefault}{\tt 0,1}}}}
\put(6167,-3136){\makebox(0,0)[lb]{\smash{\SetFigFont{10}{12.0}{\rmdefault}{\mddefault}{\updefault}{\tt 0}}}}
\put(6174,-1118){\makebox(0,0)[lb]{\smash{\SetFigFont{10}{12.0}{\rmdefault}{\mddefault}{\updefault}{\tt 1}}}}
\put(7163,-2124){\makebox(0,0)[lb]{\smash{\SetFigFont{10}{12.0}{\rmdefault}{\mddefault}{\updefault}{\tt 0,1}}}}
\put(6391,-2138){\makebox(0,0)[lb]{\smash{\SetFigFont{10}{12.0}{\rmdefault}{\mddefault}{\updefault}$k-2$}}}
\put(5910,-2266){\makebox(0,0)[lb]{\smash{\SetFigFont{10}{12.0}{\rmdefault}{\mddefault}{\updefault}{\tt 0}}}}
\put(5288,-2117){\makebox(0,0)[lb]{\smash{\SetFigFont{10}{12.0}{\rmdefault}{\mddefault}{\updefault}{\tt 0,1}}}}
\put(4375,-2243){\makebox(0,0)[lb]{\smash{\SetFigFont{10}{12.0}{\rmdefault}{\mddefault}{\updefault}{\tt 1}}}}
\put(4726,-2873){\makebox(0,0)[lb]{\smash{\SetFigFont{10}{12.0}{\rmdefault}{\mddefault}{\updefault}{\tt 0}}}}
\put(5103,-3203){\makebox(0,0)[lb]{\smash{\SetFigFont{10}{12.0}{\rmdefault}{\mddefault}{\updefault}{\tt 1}}}}
\put(5214,-3714){\makebox(0,0)[lb]{\smash{\SetFigFont{10}{12.0}{\rmdefault}{\mddefault}{\updefault}$k-2$}}}
\put(5192,-1613){\makebox(0,0)[lb]{\smash{\SetFigFont{10}{12.0}{\rmdefault}{\mddefault}{\updefault}$k+2$}}}
\put(4007,-3729){\makebox(0,0)[lb]{\smash{\SetFigFont{10}{12.0}{\rmdefault}{\mddefault}{\updefault}$k-1$}}}
\put(4118,-2670){\makebox(0,0)[lb]{\smash{\SetFigFont{10}{12.0}{\rmdefault}{\mddefault}{\updefault}$k$}}}
\put(3999,-1621){\makebox(0,0)[lb]{\smash{\SetFigFont{10}{12.0}{\rmdefault}{\mddefault}{\updefault}$k+2$}}}
\put(3257,-3707){\makebox(0,0)[lb]{\smash{\SetFigFont{10}{12.0}{\rmdefault}{\mddefault}{\updefault}{\tt 0}}}}
\end{picture}

%% file: orb9.tex
\begin{picture}(0,0)%
\epsfig{file=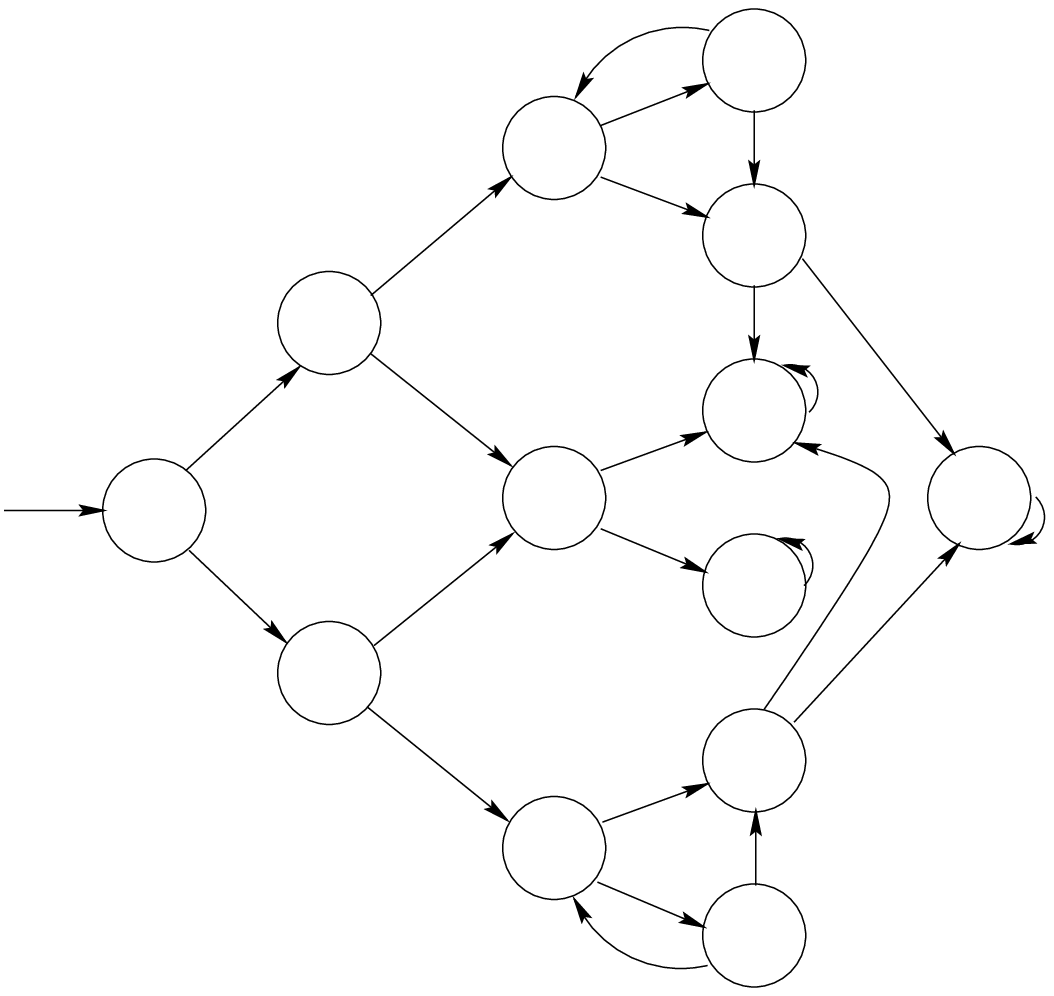}%
\end{picture}%
\setlength{\unitlength}{3158sp}%
\begingroup\makeatletter\ifx\SetFigFont\undefined%
\gdef\SetFigFont#1#2#3#4#5{%
  \reset@font\fontsize{#1}{#2pt}%
  \fontfamily{#3}\fontseries{#4}\fontshape{#5}%
  \selectfont}%
\fi\endgroup%
\begin{picture}(6352,6145)(2389,-5881)
\put(6161,-361){\makebox(0,0)[lb]{\smash{\SetFigFont{10}{12.0}{\rmdefault}{\mddefault}{\updefault}{\tt 0}}}}
\put(4781,-3301){\makebox(0,0)[lb]{\smash{\SetFigFont{10}{12.0}{\rmdefault}{\mddefault}{\updefault}{\tt 1}}}}
\put(4801,-1141){\makebox(0,0)[lb]{\smash{\SetFigFont{10}{12.0}{\rmdefault}{\mddefault}{\updefault}{\tt 1}}}}
\put(4131,-1841){\makebox(0,0)[lb]{\smash{\SetFigFont{10}{12.0}{\rmdefault}{\mddefault}{\updefault}$k+2$}}}
\put(4151,-3931){\makebox(0,0)[lb]{\smash{\SetFigFont{10}{12.0}{\rmdefault}{\mddefault}{\updefault}$k-2$}}}
\put(3091,-2961){\makebox(0,0)[lb]{\smash{\SetFigFont{10}{12.0}{\rmdefault}{\mddefault}{\updefault}$k+2$}}}
\put(3661,-2251){\makebox(0,0)[lb]{\smash{\SetFigFont{10}{12.0}{\rmdefault}{\mddefault}{\updefault}{\tt 0}}}}
\put(8011,-3471){\makebox(0,0)[lb]{\smash{\SetFigFont{10}{12.0}{\rmdefault}{\mddefault}{\updefault}{\tt 0}}}}
\put(8031,-2891){\makebox(0,0)[lb]{\smash{\SetFigFont{10}{12.0}{\rmdefault}{\mddefault}{\updefault}$k-2$}}}
\put(6681,-5541){\makebox(0,0)[lb]{\smash{\SetFigFont{10}{12.0}{\rmdefault}{\mddefault}{\updefault}$k-2$}}}
\put(6781,-251){\makebox(0,0)[lb]{\smash{\SetFigFont{10}{12.0}{\rmdefault}{\mddefault}{\updefault}$k$}}}
\put(5471,-4991){\makebox(0,0)[lb]{\smash{\SetFigFont{10}{12.0}{\rmdefault}{\mddefault}{\updefault}$k-2$}}}
\put(5451,-2891){\makebox(0,0)[lb]{\smash{\SetFigFont{10}{12.0}{\rmdefault}{\mddefault}{\updefault}$k+2$}}}
\put(5651,-801){\makebox(0,0)[lb]{\smash{\SetFigFont{10}{12.0}{\rmdefault}{\mddefault}{\updefault}$k$}}}
\put(8741,-3111){\makebox(0,0)[lb]{\smash{\SetFigFont{10}{12.0}{\rmdefault}{\mddefault}{\updefault}{\tt 0,1}}}}
\put(3511,-3531){\makebox(0,0)[lb]{\smash{\SetFigFont{10}{12.0}{\rmdefault}{\mddefault}{\updefault}{\tt 1}}}}
\put(6591,-761){\makebox(0,0)[lb]{\smash{\SetFigFont{10}{12.0}{\rmdefault}{\mddefault}{\updefault}{\tt 0}}}}
\put(6171, 99){\makebox(0,0)[lb]{\smash{\SetFigFont{10}{12.0}{\rmdefault}{\mddefault}{\updefault}{\tt 1}}}}
\put(6766,-3924){\makebox(0,0)[lb]{\smash{\SetFigFont{10}{12.0}{\rmdefault}{\mddefault}{\updefault}{\tt 1}}}}
\put(6811,-1281){\makebox(0,0)[lb]{\smash{\SetFigFont{10}{12.0}{\rmdefault}{\mddefault}{\updefault}$k$}}}
\put(6781,-2371){\makebox(0,0)[lb]{\smash{\SetFigFont{10}{12.0}{\rmdefault}{\mddefault}{\updefault}$k$}}}
\put(6711,-3411){\makebox(0,0)[lb]{\smash{\SetFigFont{10}{12.0}{\rmdefault}{\mddefault}{\updefault}$k+2$}}}
\put(5011,-4311){\makebox(0,0)[lb]{\smash{\SetFigFont{10}{12.0}{\rmdefault}{\mddefault}{\updefault}{\tt 0}}}}
\put(6111,-4501){\makebox(0,0)[lb]{\smash{\SetFigFont{10}{12.0}{\rmdefault}{\mddefault}{\updefault}{\tt 1}}}}
\put(6111,-5881){\makebox(0,0)[lb]{\smash{\SetFigFont{10}{12.0}{\rmdefault}{\mddefault}{\updefault}{\tt 1}}}}
\put(6231,-5141){\makebox(0,0)[lb]{\smash{\SetFigFont{10}{12.0}{\rmdefault}{\mddefault}{\updefault}{\tt 0}}}}
\put(6661,-4961){\makebox(0,0)[lb]{\smash{\SetFigFont{10}{12.0}{\rmdefault}{\mddefault}{\updefault}{\tt 0}}}}
\put(6681,-4461){\makebox(0,0)[lb]{\smash{\SetFigFont{10}{12.0}{\rmdefault}{\mddefault}{\updefault}$k-2$}}}
\put(7381,-2231){\makebox(0,0)[lb]{\smash{\SetFigFont{10}{12.0}{\rmdefault}{\mddefault}{\updefault}{\tt 0,1}}}}
\put(6801,-2921){\makebox(0,0)[lb]{\smash{\SetFigFont{10}{12.0}{\rmdefault}{\mddefault}{\updefault}{\tt 0,1}}}}
\put(6071,-3331){\makebox(0,0)[lb]{\smash{\SetFigFont{10}{12.0}{\rmdefault}{\mddefault}{\updefault}{\tt 0}}}}
\put(5121,-2131){\makebox(0,0)[lb]{\smash{\SetFigFont{10}{12.0}{\rmdefault}{\mddefault}{\updefault}{\tt 0}}}}
\put(6141,-2391){\makebox(0,0)[lb]{\smash{\SetFigFont{10}{12.0}{\rmdefault}{\mddefault}{\updefault}{\tt 1}}}}
\put(7351,-1481){\makebox(0,0)[lb]{\smash{\SetFigFont{10}{12.0}{\rmdefault}{\mddefault}{\updefault}{\tt 1}}}}
\put(6611,-1771){\makebox(0,0)[lb]{\smash{\SetFigFont{10}{12.0}{\rmdefault}{\mddefault}{\updefault}{\tt 0}}}}
\put(6081,-1181){\makebox(0,0)[lb]{\smash{\SetFigFont{10}{12.0}{\rmdefault}{\mddefault}{\updefault}{\tt 1}}}}
\end{picture}